\documentclass[runningheads]{llncs}

\usepackage[T1]{fontenc}
\usepackage{graphicx}
\usepackage{cite}
\usepackage[margin=3cm]{geometry}
\usepackage{algorithm2e}
\usepackage{pxfonts}
\usepackage[in]{fullpage}
\usepackage[a-1b]{pdfx}

\def\calP{\mathcal{P}}
\def\Ed{\mathsf{E}\mathrm{d}}

\begin{document}

\title{The Two-Center Problem of Uncertain Points on Cactus Graphs\thanks{This research was supported in part by U.S. National Science Foundation under Grant CCF-2339371.}}

\titlerunning{The Two-Center Problem of Uncertain Points on Cactus Graphs}

\author{Haitao Xu\and Jingru Zhang}

\authorrunning{H. Xu and J. Zhang}

\institute{Cleveland State University, Cleveland, Ohio 44115, USA\\
\email{h.xu12@vikes.csuohio.edu, j.zhang40@csuohio.edu}}

\maketitle             

\begin{abstract}
We study the two-center problem on cactus graphs in facility locations, which aims to place two facilities on the graph network to serve customers in order to minimize the maximum transportation cost. In our problem, the location of each customer is uncertain and may appear at $O(m)$ points on the network with probabilities. More specifically, given are a cactus graph $G$ and a set $\calP$ of $n$ (weighted) uncertain points where every uncertain point has $O(m)$ possible locations on $G$ each associated with a probability and is of a non-negative weight. The problem aims to compute two centers (points) on $G$ so that the maximum (weighted) expected distance of the $n$ uncertain points to their own expected closest center is minimized. No previous algorithms are known for this problem. In this paper, we present the first algorithm for this problem and it solves the problem in $O(|G|+ m^{2}n^{2}\log mn)$ time. 
\end{abstract}

\section{Introduction}\label{sec:intro}
Facility location problems with uncertain demands have been explored a lot~\cite{ref:ChauUn06,ref:HuangSt17,ref:WangAn16,ref:WangCo19, ref:AlipourAp20,ref:AlipourIm21} due to the inherent uncertainty of collected demand data. In this paper, we study a facility location problem with uncertain demands, that is, the two-center problem on a cactus graph network w.r.t. a set of $n$ uncertain points each of which has $O(m)$ possible locations on the network. 

Let $G=(V, E)$ be the given (undirected) cactus graph, which is an undirected graph where no cycles share edges. For any two vertices $u,v$, denote by $e(u,v)$ their edge and let $l(e(u,v))$ be its length. Any point $q\in e(u,v)$ is described by a tuple $(u,v,t)$ where $t$ is the distance of $q$ to $u$ along $e(u,v)$. Moreover, for any two points $p$ and $q$ on $G$, their distance $d(p,q)$ is the length of their shortest path. 

Let $\calP$ be a set of $n$ uncertain points $P_1, \cdots, P_n$. Each $P_i\in\calP$ is associated with $m$ points $p_{i1}, p_{i2},\cdots, p_{im}$ on $G$ so that $P_i$ appears at $p_{ij}$ with the probability $0\leq f_{ij}\leq 1$. Additionally, each $P_i$ has a weight $w_i\geq 0$. 

For any point $q$ on $G$, the distance between any $P_i\in\calP$ and $q$ is the \textit{expected} version, denoted by $\Ed(P_i,q)$, and it is defined as $\sum_{j=1}^{m}f_{ij}\cdot d(p_{ij}, q)$. For any two points $q_1,q_2$ on $G$, we define $\phi(q_1,q_2) = \max_{P_i\in\calP} w_i\cdot\min\{\Ed(P_i,q_1),\Ed(P_i,q_2)\}$. 
The problem aims to compute two points on $G$, that is, centers $q^*_1$ and $q^*_2$, so as to minimize $\phi(q_1,q_2)$. 

When $G$ is a tree network, this problem can be solved in $O(mn\log mn)$ time by the prune-and-search algorithm~\cite{ref:xu2023two}. On cactus graphs, if $m=1$, the problem was addressed in $O(n\log^2 n)$ time~\cite{ref:Ben-MosheEf07}. But for $m>1$, as far as we are aware, only the one-center problem has been explored~\cite{ref:HuCo23}. 
In this paper, we give the first algorithm that solves the two-center problem of uncertain points on a cactus graph in $O(|G| + m^2n^2\log mn)$ time where $|G|$ is the size of $G$, i.e., $|G| = O(|V| + |E|)$. 

\paragraph{\textbf{Related Work}} Besides the above-mentioned previous work, when $G$ is a path, this problem was addressed in $O(mn\log m + n\log n)$ time in our previous work~\cite{ref:XuTh23}. On a general graph, when $m=1$, the problem was studied by Bhattacharya and Shi~\cite{ref:BhattacharyaIm14} and they proposed an $O(|E|^2|V|\log^2 |V|)$-time algorithm. 

Another of most related problems is the general $k$-center problem where $k\geq 2$ centers are placed on the network so as to minimize the maximum (weighted) expected distance between uncertain points and centers. For $m=1$, the $k$-center problem on trees was recently addressed in $O(n\log n)$ time by Wang and Zhang~\cite{ref:WangAn21}. Bai et al.~\cite{ref:BaiTh18} gave an $O(n^2k^2)$-time algorithm for the problem on cactus graphs. When $m>1$, Wang and Zhang~\cite{ref:WangCo19} proposed an $O(n^2\log n\log mn +mn\log^2 mn\log n)$ time algorithm for this problem on tree networks. But no previous work exists for the cactus version except for $k=1$~\cite{ref:HuCo23}. 

In the plane, when $m=1$, the two-center problem under the Euclidean metric can be solved in $O(n\log^2n)$ time by Wang~\cite{ref:WangOn22}. For $m>1$, the general $k$-center problem on uncertain points is NP-hard due to the NP-hardness of the case $m=1$~\cite{ref:MegiddoOn90} and approximation algorithms have been investigated in~\cite{ref:HuangSt17,ref:AlipourIm21}. 

\paragraph{\textbf{Our Approach}} Solving this problem is equivalent to solving a vertex-constrained version where 
every location is at a vertex and every vertex holds at least one location, and the reduction can be done in linear time. Hence, we focus on solving the vertex-constrained version. 

Let $\lambda^*$ be the optimal objective value. We first need to address the decision problem, which has not been studied before. Given any value $\lambda>0$, the goal is to determine whether two centers can be placed on $G$ (to cover $\calP$) so that the objective value is no more than $\lambda$. If yes, $\lambda\geq\lambda^*$ and otherwise, $\lambda<\lambda^*$. 

Our algorithm is a generalized version of the decision algorithm~\cite{ref:xu2023two} for $G$ being a tree. Specifically, we first transform $G$ into its tree representation $T$ where each node uniquely represents a cycle, a hinge (that is a vertex on a cycle of degree more than $2$), or a $G$-vertex (that is a vertex out of any cycle). Then, at most two `binary searches' are performed on $T$ where each finds a distinct out-of-cycle edge or a cycle on $G$ that must contain a peripheral center. Last, we decide whether two centers can be placed on the obtained \textit{peripheral-center} edges or cycles to cover $\calP$ under $\lambda$. The running time is dominated by the $O(m^2n^2)$ time to handle the case where placing two centers on a cycle to cover $\calP$ under $\lambda$.   

For any point $q$ on $T$, removing it from $T$ generates several disjoint subtrees, and each of them and $q$ induce a \textit{hanging} subtree of $q$. There are two key lemmas that determine which hanging subtrees of $q$ contain the two \textit{critical} out-of-cycle edges or cycles on $G$ that contain $q^*_1$ and $q^*_2$. Based on these lemmas, we can find the critical edges or cycles on $G$ in a similar way as our decision algorithm. Once the critical edges or cycles are obtained, $\lambda^*$ can be computed in $O(m^2n^2\log mn)$ time. 

\section{Preliminaries}\label{sec:pre}
Lemma~$5$ in~\cite{ref:HuCo23} reveals that any general instance of the one-center problem w.r.t. $\calP$ on $G$ can be reduced into a vertex-constrained instance. For the two-center problem, we have the similar result. 

\begin{lemma}\label{lem:twocenterreduction}
    Any general instance of the two-center problem can be reduced to a vertex-constrained instance in $O(|G| + mn)$ time. 
\end{lemma}
\begin{proof}
    The reduction is similar to Lemma~$5$ in~\cite{ref:HuCo23}. For completeness and correctness, we include the details in the following. 
    First, we join into $G$ a vertex for every point that holds locations but is not at any vertex, and assign all locations at the point to the new vertex created for it. 
    
    We say that a subgraph of $G$ is \textit{empty} if it is free of any locations. 
    Next, starting from any out-of-cycle vertex or any hinge, we perform a depth-first search on $G$ to shrink $G$ as follows. For every cycle of only one hinge, if its every non-hinge vertex is empty, we remove this cycle from $G$ except for its only hinge. Also, we remove from $G$ every empty vertex of degree $1$ including the generated by removing such above cycles and empty vertices. Clearly, these removed components of $G$ contain no centers. 
    
    Furthermore, for each cycle with only two hinges, if its non-hinge vertices are all empty, we remove the cycle from $G$ by connecting its two hinges directly with an edge of length equal to their shortest-path length along this cycle and then pruning all its non-hinge vertices. Because it was proved in Lemma~$5$~\cite{ref:HuCo23} that the center of $\calP$ on $G$ is not interior of the longer path between two hinges in such a cycle. Clearly, this result also applies to our problem since the center that leads $\lambda^*$ is the center w.r.t. a subset of uncertain points in $\calP$. 
    
    Last, we perform a depth-first search on $G$ as the above to remove every empty vertex of degree $2$ by connecting its two adjacent vertices directly. 

    At this point, the size of $G$ is $O(mn)$ and so it contains $O(mn)$ empty vertices (each of degree at least $3$). Create $O(m)$ additional locations of zero probabilities for each $P_i\in\calP$ and assign each of them to a unique empty vertices of $G$. It follows that each $P_i\in\calP$ has $O(m)$ locations and every location is 
    at a vertex of $G$ and every vertex has at least one location. 
    The reduction takes $O(|G|+mn)$ time. It is clear to see that solving the two-center problem w.r.t. $\calP$ on $G$ is equivalent to solving the obtained vertex-constrained instance. \qed
\end{proof}

Unless otherwise specified, we assume that every location of $\calP$ is at a vertex of $G$ and every vertex contains at least one location. 

Below, we will introduce some terminologies defined in the literature~\cite{ref:BurkardAL98,ref:HuCo23}, and some important lemmas in the one-center algorithm~\cite{ref:HuCo23}, which support the correctness of our algorithm. 

A vertex in a cycle of degree at least $3$ is called a \textit{hinge}, so any cactus graph can be decomposed at its hinges into cycles and tree subgraphs called grafts. A $G$-vertex is a vertex (in a graft subgraph) that is not included in any cycle. Any cactus graph can be represented by a tree that is generated by replacing every cycle with a node and then connecting it with all its hinges by edges of length zero (see Fig.~\ref{fig:fig1}). Such a tree representation can be built in linear time~\cite{ref:BurkardAL98,ref:HuCo23}. 

Denote by $T$ the tree representation of $G$. We say a node of $T$ representing a $G$-vertex is a $G$-node, a node representing a hinge is a hinge node, and a node representing a cycle is a cycle node. On $T$, every cycle node stores all the vertex adjacency information of a cycle on $G$ as well as locations on it. Its every edge between a $G$-node and a hinge node is also in $G$. The locations at a hinge of $G$ are considered to be in its hinge nodes on $T$ rather than its copies in its cycle nodes of $T$. 

\begin{figure}[ht]
\centering
\includegraphics[width=0.6\textwidth]{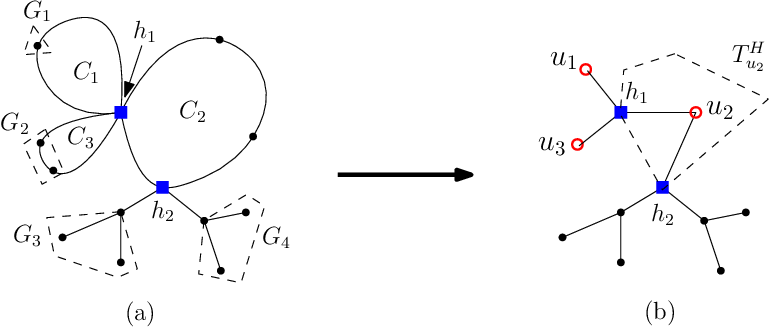}
\caption{(a) Illustrating a cactus graph that consists of three cycles $\{C_1, C_2, C_3\}$, two hinges $\{h_1, h_2\}$ (squares), and six G-vertices (disks). Cycle $C_2$ has four split subgraphs $\{G_1, G_2, G_3, G_4\}$, and four hanging subgraphs $\{C_1, C_3, G_3\cup\{h_2\}, G_4\cup\{h_2\}\}$. (b) Illustrating the tree representation of the cactus graph in (a) where each cycle is uniquely represented by a (circular) node. Specifically, node $u_1$ is for cycle $C_1$, node $u_2$ for cycle $C_2$, and node $u_3$ for cycle $C_3$. Additionally, $u_2$ and its all adjacent hinge nodes induce its $H$-subtree $T^H_{u_2}$. Every hanging subgraph of $C_2$ is uniquely represented by a hanging subtree of $T^H_{u_2}$.}  
\label{fig:fig1}
\end{figure}

A point on $G$ is called an {\em articulation} point if it is out of any cycles except for hinges. Let $q$ be any articulation point on $G$. Removing $q$ from $G$ results in multiple disjoint subgraphs. Each subgraph is called a \textit{split} subgraph of $q$, and $q$ and its every split subgraph induce a \textit{hanging} subgraph of $q$. Generally, for any subgraph $G'$ of $G$, removing $G'$ generates several split subgraphs of $G'$, and each split subgraph and its adjacent vertices on $G'$ form a hanging subgraph of $G'$. These notions also apply to $T$. 

For any $P_i\in\calP$, the sum of probabilities of all $P_i$'s positions in a subgraph $G'$ of $G$ is referred to as the \textit{probability sum} of $P_i$ in $G'$. The median of $P_i$ is a point on $G$ that minimizes $\Ed(P_i,x)$. Note that the median of $P_i$ may not be unique but let $x_i^*$ denote any of them. The following lemmas are proved in the one-center algorithm~\cite{ref:HuCo23}. 

\begin{lemma}\label{lem:articulatemedian}
\cite{ref:HuCo23} Consider any articulation point $x$ on $G$ and any uncertain point $P_i\in\calP$.
    \begin{enumerate}
        \item If $x$ has a split subgraph whose probability sum of $P_i$ is greater than $0.5$, then its median $x_i^*$ is on the hanging subgraph including that split subgraph; 
        
        \item The point $x$ is $x_i^*$ if $P_i$'s probability sum in each split subgraph of $x$ is less than $0.5$;
        
        \item The point $x$ is $x_i^*$ if $x$ has a split subgraph with $P_i$'s probability sum equal to $0.5$.      
    \end{enumerate}
\end{lemma}

Let $u$ be any cycle node of $T$. The subtree induced by $u$ and all its adjacent hinge nodes on $T$ is called the \textit{H-subtree} of $u$~\cite{ref:HuCo23}. Each hanging subtree of $u$'s $H$-subtree uniquely represents a hanging subgraph of its cycle on $G$, and vice versa. 
The following lemma is a restatement of Corollary 2 in~\cite{ref:HuCo23}.

\begin{lemma}\label{lem:cyclemedian}
\cite{ref:HuCo23} Consider any cycle node $u$ on $T$ and any uncertain point $P_i\in\calP$. 
\begin{enumerate}
    \item If the H-subtree $T^H_u$ of $u$ has a split subtree whose probability sum of $P_i$ is greater than $0.5$, then $x_i^*$ is in that split subtree; 
    
    \item If $P_i$'s probability sum in each split subtree of $T^H_u$ is less than $0.5$, then $x_i^*$ is in the cycle $u$ represents;  

    \item If $T^H_u$ has a split subtree whose probability sum of $P_i$ equals to $0.5$, then $x_i^*$ is at that hinge node of $T^H_u$ that is adjacent to the split subtree. 
\end{enumerate}
\end{lemma}

Denote by $G_u$ the subgraph of $G$ that any node $u$ on $T$ represents. For any subgraph $G'$, let $\calP^>(G')$ (resp., $\calP^{=}(G')$) be the subset of all uncertain points whose probability sums in $G'$ are greater than (equal to) $0.5$. For any node $u$ on $T$, generally, $G_u$ is assumed to have $s$ split subgraphs $G_1, \cdots, G_s$ and thereby $s$ hanging subgraphs $H_1,\cdots, H_s$. 
Denote by $h_i$ the vertex on $G_u$ that connects $G_u$ and $G_i$ or $H_i$. Clearly, if $u$ is a cycle node then $h_i$ is its hinge, and otherwise, $h_i$ is $G_u$ itself. 

As the one-center algorithm~\cite{ref:HuCo23}, we can perform all computations on $T$ instead of $G$. The following operations can be carried out on $T$ in $O(mn)$ time~\cite{ref:HuCo23}, which supports the efficiency proof of our algorithm. 

\begin{enumerate}
    \item Given any point $x$ on $G$, computing for all $1\leq i\leq n$ the expected distance $\Ed(P_i,x)$ in $O(mn)$ time;
    \item Given any node $u$ on $T$, with $O(mn)$-time preprocessing work, for any $1\leq i\leq s$ and any point $x$ of $G$, computing the distance $d(h_i, x)$ in $O(1)$ time; 
    \item Given any node $u$ on $T$, computing for each $1\leq i\leq s$ the probability sum in $G_i$ of each $P_j$ in $\calP^>(G_i)\cup\calP^=(G_i)$ and subsets $\calP^>(G_i)$ and $\calP^=(G_i)$ in $O(mn)$ time. 
\end{enumerate}

\section{The Decision Algorithm}\label{sec:decision}
Given any value $\lambda>0$, the decision problem is to determine whether two centers can be placed on $G$ so that the objective value is no more than $\lambda$. If yes, $\lambda\geq\lambda^*$ and it is feasible; otherwise, $\lambda<\lambda^*$ so it is infeasible. For any point $x$ of $G$, we say that $x$ covers a subset $\calP'$ of $\calP$ (under $\lambda$) iff $\max_{P_i\in\calP'}w_i\cdot Ed(P_i,x)\leq\lambda$. In other words, the decision problem aims to find at most two points (to place centers) on $G$ that cover $\calP$ completely.  

We say that an edge $e$ of $G$ is a \textit{peripheral-center} edge (w.r.t. $\lambda$) if it satisfies the following conditions: (1) $e$ is not in any cycle; (2) $e$ has a hanging subgraph, say $G'$, so that the subset of all uncertain points that have their probability sums in $G'$ larger than $0.5$ can be covered by a point of $e$; (3) an uncertain point whose median is in $G'$ cannot be covered by any point in $G/\{G'\cup e\}$. Lemma~\ref{lem:articulatemedian} implies that for each $P_i$ whose median is in $G'$, $\Ed(P_i,x)$ is monotonically increasing as $x$ moves along $e$ away from $G'$ and $\Ed(P_i,x)$ at any point of $G/\{G'\cup e\}$ is not less than its value at $e$'s incident vertex not in $G'$. Hence, in order to place least centers to cover $\calP$, a center must be placed on such a peripheral-center edge $e$ to cover all uncertain points with medians in $G'$ and most uncertain points with medians in $G/G'$; such a center on it is called a \textit{peripheral} center (in that all other necessary centers are placed in its hanging subgraph excluding $G'$). 



Clearly, peripheral-center edges are ``outermost'' subgraphs of smallest size on $G$ that must contain (peripheral) centers under $\lambda$. Hence, the decision algorithm~\cite{ref:XuTh23} for $G$ being a tree performs ``binary searches'' on $G$ to find two peripheral-center edges, and then determines whether two points can be found on the peripheral-center edges to cover $\calP$; if yes, $\lambda$ is feasible and otherwise, more than two centers are needed so $\lambda$ is infeasible. 

In general, $G$ contains cycles. Besides of peripheral-center edges, $G$ might have ``outermost'' cycles that must contain centers under $\lambda$. Similarly, a cycle $C$ of $G$ is a \textit{peripheral-center} cycle under $\lambda$ if it satisfies the following conditions: (1) $C$ has at most one hanging subgraph, say $G''$, so that the subset of uncertain points that have probability sums in $G''$ larger than $0.5$ cannot be covered by a point of $C$; (2) an uncertain point with its median in $C\cup\{G/G''\}$ cannot be covered by any point of $G''/\{G''\cap C\}$, where $G''\cap C$ is their common vertex. Because each uncertain point with its median in $C\cup\{G/G''\}$ has its expected distance at any point of $G''/\{G''\cap C\}$ not less than its expected distance at vertex $G''\cap C$. In order to place least centers covering $\calP$, one must place center(s) in such a peripheral-center cycle $C$ to cover the subset of uncertain points satisfying condition (2) and most uncertain points not in this subset, and other necessary centers must be placed in $C\cup G''$.

Our algorithm is a generalization of the decision algorithm~\cite{ref:XuTh23} for trees: First, it finds on the tree representation $T$ of $G$ two smallest subtrees that each represents a peripheral-center edge or cycle on $G$ under $\lambda$. Next, we determine if two centers can be placed on the obtained peripheral-center edge(s) and cycle(s) to cover $\calP$ under $\lambda$. If there exist two points on them that can cover $\calP$, then $\lambda\geq\lambda^*$, and otherwise, more than two points are needed to cover $\calP$, so $\lambda<\lambda^*$. 


Let $c^1$ and $c^2$ represent two arbitrary peripheral-center edges or cycles on $G$. Note that $c^1$ is $c^2$ if all necessary centers covering $\calP$ under $\lambda$ are in the same cycle or out-of-cycle edge. Let $T^1$ and $T^2$ represent the two subtrees of $T$ that respectively consist of $c^1$ and $c^2$ (i.e., their representations on $T$). Initially, we set $T^1 = T^2 = T$, and then assign a flag $f$ to every node $u$ of $T^1$ (resp., $T^2$) so that $u$'s $f$ is true iff the hanging subtrees of $u$ in $T/T^1$ (resp., $T/T^2$) must contain center(s) under $\lambda$. The flag $f$ of every node in $T^1$ (resp., $T^2$) is initialized as false. 

We shall now present how to find a peripheral-center edge or cycle node on $T^1$ and $T^2$ in turn. This routine needs the supports of Lemmas~\ref{lem:centerdetectingarticulate} and~\ref{lem:centerdetectingcycle} whose proofs are delayed to Subsection~\ref{sec:lemmas}. 

Compute the \textit{centroid} $c$ on $T^1$ which is a node of $T^1$ whose hanging subtrees each is of size no more than half of $T^1$ 
and can be found in $O(|T^1|)$ time~\cite{ref:KarivAn79, ref:KeikhaCl21}. Then, if $c$ is a hinge or $G$-node, by applying Lemma~\ref{lem:centerdetectingarticulate} to $c$, we determine 
whether a center must be placed at $G_c$ and if not, which hanging subtree of $c$ contains a peripheral-center edge or cycle. If only $c$ is returned, which means that $G_c$ must be a center, then we set $T^1$ as $c$.  

Generally, Lemma~\ref{lem:centerdetectingarticulate} returns one or two nodes that are adjacent to $c$ in $T^1$. In the former case where only one adjacent node of $c$ is returned, the corresponding hanging subtree of $c$ on $T$ contains all necessary centers. Hence, we find in $T^1$ $c$'s hanging subtree including this node and reset $T^1$ as this subtree. The latter case means that the two corresponding hanging subtrees of $c$ on $T$ each contains a peripheral-center edge or cycle. So, we find in $O(|T^1|)$ time the two hanging subtrees of $c$ in $T^1$ respectively containing the two obtained nodes. Set $T^1$ as the one of them that excludes any node with a true flag. Last we set $c$'s flag $f$ in $T^1$ as true.

On the other hand, $c$ is a cycle node. Instead, we apply Lemma~\ref{lem:centerdetectingcycle} to decide whether $G_c$ is a peripheral-center cycle and which hanging subtree of $c$'s $H$-subtree $T^H_c$ in $T^1$ contains a peripheral-center edge or cycle. As described in the proof for Lemma~\ref{lem:centerdetectingcycle}, if $c$ is obtained, which means that $G_c$ is a peripheral-center cycle, 
then we let $T^1$ be $c$. If one of its adjacent hinge nodes is returned, which can be known in $O(|T^1|)$ time, then a center must be placed at that hinge; so we set $T^1$ as that hinge node. 

Otherwise, we obtain two nodes that are respectively in two split subtrees of $T^H_c$ and adjacent to $c$'s hinge node(s), or obtain only one such node. For the former situation, we find in linear time a hanging subtree of $T^H_c$ in $T^1$ that includes one of the two nodes but excludes any node with a true flag $f$; next we set $T^1$ as this subtree and set the flag $f$ as true of $c$'s hinge node in $T^1$. For the latter case, we find the hanging subtree of $T^H_c$ in $T^1$ that contains the only obtained node and reset $T^1$ as this subtree. 

We recursively apply the above procedure to $T^1$ to find a peripheral-center edge or cycle until $T^1$ consists of only one node or one edge. Clearly, it takes $O(\log mn)$ recursive steps, so the running time is $O(mn\log mn)$. 

Further, depending on whether $T^1$ is an edge or a node on $T$, we deal with each of the following cases. 

\begin{enumerate}
    \item $T^1$ is a hinge node or $G$-node. A center must be placed at vertex $G_{T^1}$ on $G$, which is represented by $T^1$. We can decide the feasibility of $\lambda$ in $O(mn\log mn)$ time as follows. Set the weight of every uncertain point that can be covered by $G_{T^1}$ under $\lambda$ as zero. Then we apply the one-center algorithm~\cite{ref:HuCo23} to solve the one-center problem of $G$ w.r.t. $\calP$ in $O(mn\log mn)$ time. If the obtained objective value is not greater than $\lambda$, then $G_{T^1}$ and the obtained center can cover $\calP$ completely under $\lambda$, which means that $\lambda$ is feasible. Otherwise, $\lambda$ is not feasible. 
    
    \item $T^1$ is a cycle node. By the definition, we first
    determine which hanging subtree of $T^1$'s $H$-subtree contains a peripheral-center edge or cycle. If none, then all necessary centers can be placed on $G_{T^1}$. Hence, Lemma~\ref{lem:dec-one-cycle} presented in Subsection~\ref{sec:lemmas} is called to decide the feasibility of $\lambda$ in $O(m^2n^2)$ time. Otherwise, we let $T^2$ be that hanging subtree of $T^1$'s $H$-subtree in $T$ and set the flag $f$ as true for the hinge node of $T^2$ that connects $T^1$ and $T^2$ in $T$, which can be carried out in $O(|T|)$ time. 
    
    \item $T^1$ is an edge. If it is between a hinge and a cycle node, clearly, this cycle node represents a peripheral-center cycle. We set $T^1$ as this cycle node and then handle this case in the above way in $O(m^2n^2)$ time. Otherwise, $T^1$ is an out-of-cycle edge on $G$. We call Lemma~\ref{lem:dec-one-edge} in Subsection~\ref{sec:lemmas} to decide the feasibility of $\lambda$ in $O(mn\log mn)$ time. 
\end{enumerate}

In general, $T^1$ is a cycle node representing a peripheral-center cycle, i.e., $c^1$, and $T^2$ is a subtree of $T$ which contains no vertices in $T^1$ and where a leaf connects it and $T^1$ on $T$ and is of flag $f$ being true. $G_{T^2}$ contains all other necessary centers under $\lambda$. 

Proceed with finding a peripheral-center edge or cycle on $T^2$ recursively in the above way. After $O(\log mn)$ recursive steps, $T^2$ shrinks down to a node or an edge on $T$. If it is an out-of-cycle edge, we call Lemma~\ref{lem:dec-one-edge} to decide the feasibility of $\lambda$ in $O(mn\log mn)$ time. Otherwise, $T^2$ contains exactly one cycle node, which is not same as $T^1$; $c^2$ is the cycle on $G$ that $T^2$ represents. Subsequently, Lemma~\ref{lem:dec-two-cycle} in Subsection~\ref{sec:lemmas} is applied to $T^1$ and $T^2$, which decides in $O(mn^2)$ time the feasibility of $\lambda$. 

Regarding the time complexity, in the  worst case, cycle $c^1$ (i.e., $G_{T^1}$) contains all necessary centers so that the feasibility of $\lambda$ is decided in $O(m^2n^2)$ time by Lemma~\ref{lem:dec-one-cycle}, 
which dominates the time to find $c^1$. Hence, our algorithm decides the feasibility of any given $\lambda$ in $O(m^2n^2)$ time. 

\begin{lemma}\label{lem:decision}
    The decision problem can be solved in $O(m^2n^2)$ time. 
\end{lemma}

\subsection{Proofs of Lemmas}\label{sec:lemmas}
We present the proofs of all Lemmas for solving the following problems raised in our decision algorithm. 

To present our proofs, some notations are introduced first. For any subset $\calP'\in\calP$ and any point $x$ on $G$, define $\Lambda(\calP',x) = \max_{P_i\in\calP'}w_i\Ed(P_i,x)$. Recall that for any subgraph $G'\in G$, $\calP^>(G')$ (resp., $\calP^=(G')$) is the subset of all uncertain points whose probability sums in $G'$ are greater than (resp., equal to) $0.5$. 

The first problem aims to decide for any given hinge node or $G$-node $u$, whether a center must be placed at vertex $G_u$, and which hanging subtrees of $u$ in $T$ must contain peripheral-center edges or cycles. The following lemma gives the result. 

\begin{lemma}\label{lem:centerdetectingarticulate}
    Given a hinge or $G$-node $u$ on $T$, we can decide in $O(mn)$ time whether a center must be placed at vertex $G_u$, and if not, which hanging subtree of $u$ contains a peripheral-center edge or cycle. 
\end{lemma}
\begin{proof}
    For each split subgraph $G_i$ of $G_u$, let $\calP_i$ be the subset of all uncertain points whose probability sums in $G_i$ are at least $0.5$ but less than $0.5$ in every other split subgraph. Let $\calP'$ be the set of all uncertain points whose probability sums in $G_u$'s each split subgraph are less than $0.5$, and $\calP''$ denotes the set of uncertain points that occur in two such subsets $\calP^=(G_i)$. Implied by Lemma~\ref{lem:articulatemedian}, each $P_j\in\calP_i$ has its median in $H_i$ but not in any other hanging subgraph except for $G_u$; each $P_i\in\calP'$ has its median (only) at $G_u$; each $P_i\in\calP''$ has its median at $G_u$ but may also achieve its minimum expected distance at other points in the two split subgraphs where its probability sums are $0.5$. 
     
    Lemmas~\ref{lem:articulatemedian} and~\ref{lem:cyclemedian} imply the following properties. If $\Lambda(\calP', G_u)>\lambda$ or $\Lambda(\calP'', G_u)>\lambda$ then $\lambda$ is infeasible due to $\lambda^*\geq\max_{P_k\in\calP}w_k\Ed(P_k, x^*_k)$. When $\Lambda(\calP', G_u)=\lambda$ or $\Lambda(\calP_j, G_u) = \lambda$ for some $1\leq j\leq s$, a center must be placed at $G_u$ to cover $\calP'$ or $\calP_j$ and most uncertain points in $\calP- \calP_j$. Moreover, for any $1\leq i\leq s$, if $\Lambda(\calP_i, G_u)>\lambda$, center(s) must be placed in subgraph $H_i/\{h_i\}$ for covering $\calP_i$, so $H_i$ contains peripheral-center edge(s) or cycle(s); otherwise, it is not necessary to place a center in $H_i/\{h_i\}$ if $\Lambda(\calP_i, G_u)<\lambda$. 
    
    Recall that for all $1\leq i\leq s$, subsets $\calP^>(G_i)$ and $\calP^=(G_i)$ can be computed in $O(mn)$ time. Because subset $\calP_i$ is exactly the subset $\calP^>(G_i)\cup\calP^=(G_i)$ excluding all uncertain points that appear in $\calP^=(G_i)$ and other subset $\calP^=(G_j)$ with $1\leq i\neq j\leq s$. Every uncertain point has its probability sum equal to $0.5$ in at most two split subgraphs of $G_u$. Hence, subsets $\calP_i$ for all $1\leq i\leq s$ can be obtained in $O(mn)$ time. 
    
    Additionally, due to $\calP' = \calP-\cup_{i=1}^{s}\{\calP^>(G_i)\cup\calP^=(G_i)\}$, $\calP'$ can be known in $O(n)$ time. So is $\calP''$ on account of $\calP'' = \calP - \calP'-\cup_{i=1}^{s}\calP_i$. Because computing $\Ed(P_k, G_u)$ for all $1\leq k\leq n$ takes in $O(mn)$ time~\cite{ref:HuCo23}. It then follows that values $\Lambda(\calP', G_u)$, $\Lambda(\calP'', G_u)$, and $\Lambda(\calP_i, G_u)$ for each $1\leq i\leq s$ can be obtained in $O(mn)$ time. 

    Further, based on the above properties, we deal with each of the following cases accordingly. If $\Lambda(\calP', G_u)>\lambda$, $\Lambda(\calP'', G_u)>\lambda$, or more than two split subtrees have $\Lambda(\calP_i, G_u)>\lambda$, then $\lambda$ is not feasible. When $\Lambda(\calP', G_u)=\lambda$ or some split subgraph $G_j$ has $\Lambda(\calP_j, G_u)=\lambda$, a center must be placed at $G_u$ so we return node $u$. When exactly two split subgraphs have $\Lambda(\calP_i, G_u)>\lambda$, we return the two adjacent nodes of $u$ in its two hanging subtrees respectively representing the two split subgraphs. Similarly, for the case where only one split subgraph has $\Lambda(\calP_i, G_u) >\lambda$, we return $u$'s adjacent node in the corresponding split subtree. Otherwise, $\Lambda(\calP', G_u)<\lambda$, $\Lambda(\calP'', G_u)\leq\lambda$, and $\Lambda(\calP_i, G_u)<\lambda$ for each $1\leq i\leq s$. Clearly, a center at $G_u$ can cover all uncertain points in $\calP$ under $\lambda$, so $\lambda$ is feasible.
    
    It is clear to see that the total running time is $O(mn)$. Thus, the lemma is proved. \qed
\end{proof}

Let $u$ be any cycle node on $T$. The second problem to address is to decide whether $G_u$ is a peripheral-center cycle and which hanging subtree of $u$'s $H$-subtree $T^H_u$ contains a peripheral-center edge or cycle.  

\begin{lemma}\label{lem:centerdetectingcycle}
    For any cycle node $u$ on $T$, we can decide in $O(mn)$ time whether $G_u$ is a peripheral-center cycle and which hanging subtree of $T^H_u$ contains a peripheral-center edge or cycle. 
\end{lemma}
\begin{proof}
    Similarly, we can compute the subset $\calP_i$ for every split subgraph $G_i$ of $T^H_u$ in $O(mn)$ time. Values $\Lambda(\calP_i, h_i)$ for all $1\leq i\leq s$ are obtained in the same time complexity since for any point $x$, $d(x,h_i)$ with $1\leq i\leq s$ can be known in $O(1)$ time after an $O(mn)$-time preprocessing work. 
    
    By Lemma~\ref{lem:articulatemedian} and Lemma~\ref{lem:cyclemedian}, we have the following observation. For any $1\leq i\leq s$, if $\Lambda(\calP_i, h_i)>\lambda$ then the hanging subgraph $H_i$ must contain a peripheral-center edge or cycle; if $\Lambda(\calP_i, h_i) = \lambda$ then a center must be placed at hinge $h_i$ to cover most uncertain points. Otherwise, $H_i$ does not contain any peripheral-center edge 
    or cycle.  
    
    Now we can handle each case as follows: If any split subgraph $G_i$ has $\Lambda(\calP_i, h_i)=\lambda$, then we return only the hinge node of $h_i$ on $T$, which can be found in $O(1)$ time. When more than two split subgraphs have $\Lambda(\calP_i, h_i)>\lambda$, $\lambda$ is not feasible. When exactly two split subgraphs of $G_u$ have $\Lambda(\calP_i, h_i)>\lambda$, which means $G_u$ is not a peripheral-center cycle, we return the two nodes adjacent to $u$'s hinge nodes respectively in the two corresponding split subtrees of $T^H_u$, which can be obtained in $O(1)$ time. 
  
    For the case where only one split subgraph, e.g., $G_i$, has $\Lambda(\calP_i, h_i)>\lambda$, we apply Lemma~\ref{lem:centerdetectingarticulate} to $h_i$ in order to decide whether $G_u$ is a peripheral-center cycle. In general, we obtain only node $u$, the node $u'$ that is adjacent to the hinge node of $h_i$ in $T^H_u$'s split subtree for $G_i$, or both of them. By Lemma~\ref{lem:centerdetectingarticulate}, the first case means that $G_u$ contains all necessary centers so we return only $u$. For the second case, all necessary centers are in $H_i$ so $u'$ is returned. In the last case, $G_u$ is a peripheral-center cycle and $H_i$ contains another peripheral-center edge or cycle, so we return both $u$ and $u'$. 
    
    Clearly, each case is handled in $O(mn)$ time in worst case. Hence, the time complexity is $O(mn)$. The lemma thus holds. \qed
\end{proof}

In the decision algorithm, we need to decide the feasibility of $\lambda$ when a peripheral-center edge, a peripheral-center cycle that contains all necessary centers, 
or two peripheral-center cycles are obtained. These cases are handled by the following lemmas. 

\begin{lemma}\label{lem:dec-one-edge}
Given a peripheral-center edge, we can determine in $O(mn\log mn)$ time whether $\lambda\geq\lambda^*$.
\end{lemma}
\begin{proof}
Let $e$ be the given peripheral-center (out-of-cycle) edge. 
Assume that it is incident to vertices $v_1$ and $v_2$ on $T$, which are non-cycle nodes. $e$ has exactly two hanging subgraphs $H_1, H_2$. Suppose $v_1$ is in $H_1$ and $v_2$ is in $H_2$. By the definition, a center must be placed on $e$ to cover all uncertain points that have medians all in $H_1$ or $H_2$. 

Recall that subset $\calP^>(H_1)\cup\calP^=(H_1)$ (resp., $\calP^>(H_2)\cup\calP^=(H_2)$) is the subset of all uncertain points whose medians are in $H_1$ (resp., $H_2$), which can be obtained in $O(mn)$ time. For convenience, let $\calP_1 = \calP^>(H_1)\cup\calP^=(H_1)$ and $\calP_2=\calP^>(H_2)\cup\calP^=(H_2)$. 

To place necessary center(s) on $e$, we first find the point on $e$ furthest to $v_1$ that covers subset $\calP_1$. If it exists, then we set the weights of all uncertain points covered by it as zero and subsequently solve the one-center problem w.r.t. $\calP$ on $G$. 
Also, we compute such a point on $e$ covering subset $\calP_2$ and if it exists, solve the one-center problem on $G$ w.r.t. the subset of all uncovered uncertain points. 
As long as one obtained one-center objective value is no greater than $\lambda$, $\lambda$ is feasible. Otherwise, $\lambda<\lambda^*$.  

It remains to discuss how to compute such a point on $e$ to cover $\calP_1$ (resp., $\calP_2$). 
Let $x$ be any point on $e$. 
For each $P_i$ in $\calP_1$ (resp., $\calP_2$), 
we compute the furthest point on $e$ to $v_1$ (resp., $v_2$) that covers $P_i$ by solving $w_i\Ed(P_i,x)\leq\lambda$ for $x\in e$. If all such points exist, the necessary center on $e$ must be placed at the one of them that is closest to $v_1$ (resp., $v_2$). 

Let $F_i$ denote the probability sum of each $P_i\in\calP$ in $H_1$. Clearly, for any $x\in e$, $\Ed(P_i,x) = \Ed(P_i,v_1) + (2F_i-1)\cdot d(x,v_1)$. Recall that all $\Ed(P_i,v_1)$ can be computed in $O(mn)$ time. Additionally, values $F_i$ for all $1\leq i\leq n$ can be obtained in linear time by traversing the hanging subtree of $e$ representing $H_1$. 
Hence, solving $w_i\Ed(P_i,x)\leq\lambda$ for each $P_i\in\calP$ takes a constant time. 
As a result, the furthest  point on $e$ to $v_1$ (resp., $v_2$) that covers $\calP_1$ (resp., $\calP_2$) can be computed in $O(mn)$ time. 

Since solving the one-center problem takes $O(mn\log mn)$ time~\cite{ref:HuCo23}, the time complexity is $O(mn\log mn)$. Thus, the lemma holds. \qed
\end{proof}

\begin{lemma}\label{lem:dec-one-cycle}
Given a peripheral-center cycle node, we can determine in $O(m^2n^2)$ 
time whether $\lambda\geq\lambda^*$. 
\end{lemma}
\begin{proof}
Let $u$ be the given node on $T$. The goal is to decide if two centers can be placed on cycle 
$G_u$ to cover $\calP$ under $\lambda$. 

Suppose $v_1, v_2, \cdots, v_s$ are the clockwise enumeration of all 
vertices on $G_u$. Let $x$ be any point on $G_u$. 
For each $P_i\in\calP$, consider function $y = w_i\Ed(P_i,x)$ 
in the $x,y$-coordinate plane by setting vertices $v_1, v_2, \cdots, v_s, v_1$ on the $x$-axis in order. As analyzed in~\cite{ref:HuCo23}, $\Ed(P_i,x)$ 
is a piecewise linear function of complexity $O(m)$. Hence, solving $w_i\Ed(P_i,x)\leq \lambda$ generates a set $S^i$ 
of $O(m)$ disjoint line segments $s^i_1, s^i_2, \cdots, s^i_z$ on the $x$-axis, that is, $O(m)$ disjoint arcs on $G_u$. 

For all $1\leq i\leq n$, functions $y = w_i\Ed(P_i,x)$ in $x\in G_u$ can be determined in $O(mn\log mn)$ time~\cite{ref:HuCo23}. Because each $y = \Ed(P_i,x)$ is piecewise linear and of complexity $O(m)$. It follows that all $n$ (ordered) sets $S^i$, generated by solving $w_i\Ed(P_i,x)\leq \lambda$, can be built in $O(mn\log mn)$ time. 

To decide the feasibility of $\lambda$, it is equivalent to determining whether there exist two endpoints among all of line segments in the $n$ sets $S^i$ so that a line segment of every set $S^i$ contains one of them, that is, each set $S^i$ is hit by one of them. The sweeping technique can be adapted to address this problem as follows. 

Sort all endpoints of sets $S^i$ from left to right on $x$-axis with an additional requirement where left endpoints precede right endpoints in the order if they are of same $x$-coordinate. Let $X = \{x_1, \cdots,  x_M\}$ denote this order where $M = O(mn)$. $X$ can be obtained in $O(mn\log  mn)$ time. 

Let $l$ be a vertical line in the $x,y$-coordinate plane, which will be used to sweep all line segments on $x$-axis from $x_1$ to $x_M$. Create an auxiliary array $F[1\cdots n]$ so that if set $S^i$ is hit by the endpoint that  $l$ meets, then $F[i]$ is set as one, and otherwise, $F[i]$ is set as zero. All entries of $F$ are initialized as zero. Additionally, we set a counter $\alpha$ to record the number of sets $S^i$ that $l$ hits. 

The algorithm consists of at most $M$ iterations where the $i$-th iteration determines whether there is an endpoint $x_j$ with $1\leq i\leq j\leq M$ so that endpoints $x_i, x_j$ hit all sets $S^i$. Specifically, in the $i$-th iteration, we begin by setting $\alpha = 0$ and reset every entry of $F$ as zero. Sweep $l$ to scan $X$ from $x_1$ to $x_M$ until $l$ meets a right endpoint whose $x$-coordinate equals to $x_i$'s. 
For each encountered endpoint of $x$-coordinate smaller than $x_i$'s, supposing it belongs to set $S^k$, if it is a left endpoint, then $F[k]$ is set as one and $\alpha$ is incremented; if it is its right endpoint then we set $F[k]$ as zero and decrement $\alpha$ (since all segments in each set are disjoint). 
In $O(mn)$ time, $l$ reaches the first right endpoint of $x$-coordinate equal to $x_i$. At this moment, the number of sets that $x_i$ hits equals to $\alpha$, and enumerating indices of all `one' entries in $F$ generates this subset. 

Further, we check if $(n-\alpha)$ equals to zero. If yes then $x_i$ hits all sets so $\lambda$ is feasible. Otherwise, set $\beta = n-\alpha$ and reset $\alpha$ as zero. Proceed with sweeping the remaining endpoints to find an endpoint that hits all remaining segment sets. 
For each endpoint encountered, supposing it belongs to $S^j$, 
if it is its left endpoint and $F[j] = 0$, we increment $\alpha$ and then  return $\lambda\geq\lambda^*$ when $\alpha\geq\beta$. 
Otherwise, it is its right endpoint and we decrement $\alpha$ if $F[j] = 0$. 

It is clear to see that every iteration takes $O(mn)$ time to find the existence of such an endpoint that hit all sets with a given endpoint. 
At the end, we return $\lambda<\lambda^*$ if no two endpoints in $X$ can hit all sets $S^i$. Since the time of the $O(M)$ iterations 
dominates, the time complexity is $O(m^2n^2)$.\qed
\end{proof}

\begin{lemma}\label{lem:dec-two-cycle}
Given two peripheral-center cycles, we can determine in $O(mn^2)$ time whether $\lambda\geq\lambda^*$.  
\end{lemma}
\begin{proof}
Let $u_1$ and $u_2$ be the two given peripheral-center cycle nodes on $T$. Cycle $G_{u_2}$ (resp., $G_{u_1}$) belongs to a hanging subgraph of $G_{u_1}$ (resp., $G_{u_2}$) and it connects $G_{u_1}$ (resp., $G_{u_2}$) by the hinge $h_1$ (resp., $h_2$) on $G_{u_1}$ (resp., $G_{u_2}$). Denote by $G'$ the subgraph generated by removing from $G$ the split subgraph of $G_{u_1}$ intersecting $G_{u_2}$. 

By the definition, at least one center must be placed on $G_{u_1}$ to cover all uncertain points whose medians are in $G'$ but that cannot be covered by any points out of $G_{u_1}$. Recall that subset $\calP' = \calP^>(G')\cup\calP^=(G')$ is the subset of all uncertain points whose medians are in $G'$, i.e., whose probability sums on $G'$ are no less than $0.5$.  
Denote by $\calP''$ the subset of $\calP'$ where all uncertain points can be covered by $h_2\in G_{u_2}$. Clearly, subsets $\calP'$ and $\calP''$ can be obtained in $O(mn)$ time. 

For each $P_i\in\calP''$, solving $w_i\Ed(P_i,x)\leq\lambda$ for $x\in G_{u_2}$ generates an arc on $G_{u_2}$ that contains $h_2$ or only point $h_2$. Consider $h_2$ as an infinitesimal arc on $G_{u_2}$. Let $S = \{S_1, \cdots, S_t\}$ be the set of all (disjoint) arcs defined by all endpoints of these obtained arcs on $G_{u_2}$ for $\calP''$. Set $S$ contains the infinitesimal arc $h_2$ if some uncertain point in $\calP''$ is covered by only $h_2\in G_{u_2}$. Denote by $\calP''(S_i)$ the subset of all uncertain points in $\calP''$ whose arcs contain $S_i$ entirely.  

Suppose that there exist two points respectively on $G_{u_1}$ and $G_{u_2}$ that cover $\calP$ under $\lambda$. 
Depending on which of subsets $\calP''(S_i)$ is covered by the center on $G_{u_2}$, for each $1\leq i\leq t$, we have only two options for placing the center on $G_{u_1}$ in order to place least centers on $G$ to cover $\calP$: 
They are the two points on $G_{u_1}$ respectively clockwise and counterclockwise closest to $h_1$ and that each covers subset $\calP'-\calP''(S_i)$. Denote by $X$ the set of all these $O(n)$ optional points for placing the center on $G_{u_1}$. 

Clearly, for any point in $X$, if a point on $G_{u_2}$ can cover all remaining uncertain points that cannot be covered by it, then our assumption is true and so $\lambda$ is feasible. Otherwise, our assumption is false. Additionally, some uncertain points in $\calP'$ cannot be covered by any point out of $G'/\{h_1\}$ whereas some in $\calP-\calP'$ cannot be covered by any point in $G'$. Hence, two centers are not enough to cover $\calP$ under $\lambda$, which means $\lambda<\lambda^*$. 

All $n$ functions $y = w_i\Ed(P_i,x)$ for $x\in G_{u_1}$ (resp., $x\in G_{u_2}$) can be determined in $O(mn)$ time. It follows that computing the set $S$ takes $O(mn)$ time and all subsets $\calP''(S_i)$ are obtained in $O(n^2)$ time. 
Although all subsets $\calP''(S_i)$ might be implicitly formed in $O(n)$ time since $h_2$ is in every arc that leads $S$, 
the following computation leads the time complexity at least $O(n^2)$. 

We compute for each $1\leq i\leq t$ the two above-mentioned 
optional points to place a  center on $G_{u_1}$, i.e., the set $X$. For each $P_j$ in subset $\calP'$, we solve in $O(m)$ time $w_j\Ed(P_j,x)\leq\lambda$ for $x\in G_{u_1}$, which generates a set of at most $m$ disjoint arcs on $G_{u_1}$ in order. Sort all endpoints in both clockwise and counterclockwise directions starting from $h_1$. Next, for each $1\leq i\leq t$, we scan all endpoints to find the clockwise and counterclockwise endpoints to $h_1$ that hit each of the obtained arc sets for all uncertain points in subset $\calP' - \calP''(S_i)$. Clearly, this can be carried out in $O(mn^2)$ time by the sweeping technique as in Lemma~\ref{lem:dec-one-cycle}. 
Hence, the set $X$ of all $O(n)$ optional points for placing a center on $G_{u_1}$ can be obtained in $O(mn^2 + mn\log mn)$ time. 

It should be noticed that if such two optional points for every subset $\calP''(S_i)$ do not exist then $\lambda$ is not feasible. 

Last, we decide for each point on $G_{u_1}$ in $X$, whether there exists a point on $G_{u_2}$ that covers all uncertain points that it cannot cover. For every point of $X$, find all uncertain points that it covers under $\lambda$ in $O(mn^2)$ time. Solve $w_i\Ed(P_i,x)\leq\lambda$ for $x\in G_{u_2}$ for each $P_i\in\calP$ and then sort the obtained endpoints. Subsequently, for each point of $X$, we decide the existence of a point on $G_{u_2}$ to cover the uncovered uncertain point by sweeping all obtained arcs on $G_{u_2}$. If such a point on $G_{u_2}$ exists for some point in $X$ then $\lambda$ is feasible and so we immediately return. Otherwise, none of them have such a point on $G_{u_2}$ and thereby $\lambda$ is not feasible. The running time is $O(mn^2+mn\log mn)$. 

The total running time is dominated by the last two steps 
each of which takes $O(mn^2 +mn\log mn)$ time. Thus, the lemma is proved. \qed
\end{proof}

\section{Computing $\lambda^*$}\label{sec:alg}
In this section, we present our algorithm that computes $\lambda^*$. By the definitions, the cycle or out-of-cycle edge on $G$ that contains centers $q^*_1$ or $q^*_2$ is a peripheral-center edge or cycle on $G$ under $\lambda^*$. 
We call the cycle or the out-of-cycle edge containing $q^*_1$ or $q^*_2$ a \textit{critical} node or edge. Denote by $c^*_1$ and $c^*_2$ the critical edges or cycles.  

Similar to the decision algorithm, we perform two ``binary searches'' on $T$ to respectively find $c^*_1$ and $c^*_2$ with the support of the following Lemmas~\ref{lem:CriticalNonCycleNode} and~\ref{lem:CriticalCycleNode}. 

In short, to find $c^*_1$, we recursively compute the centroid $c$ of $T$, apply Lemmas~\ref{lem:CriticalNonCycleNode} or~\ref{lem:CriticalCycleNode} to it accordingly, and prune all its split subtrees that are not relevant to $c^*_1$. 
After $O(\log mn)$ recursive steps, in general, we obtain an out-of-cycle edge, an edge incident to a cycle node, or a cycle node. For the middle case, that cycle node is exactly the critical node sought. After then, $c^*_2$ can be found in a similar way by recursively performing the searching on the hanging subtree of the centroid containing $c^*_2$. Once $c^*_1$ and $c^*_2$ are found, Lemma~\ref{lem:basecase} are applied to compute $\lambda^*$ with assistance of our decision algorithm. 

\begin{lemma}\label{lem:CriticalNonCycleNode}
    Given a hinge node or $G$-node $u$ on $T$, we can decide in $O(m^2n^2)$ time which hanging subtrees of $u$ contain $c^*_1$ and $c^*_2$. 
\end{lemma}
\begin{proof}
    Those notations defined in Lemma~\ref{lem:centerdetectingarticulate} are adapted for $G_u$ and its $s$ split subgraphs. Compute value $y_i = w_i\Ed(P_i,G_u)$ for each $P_i\in\calP$. 
    For its each split subgraph $G_i$ with $1\leq i\leq s$, find subsets $\calP^>(G_i)$ and $\calP^=(G_i)$ as well as value $\Lambda(\calP^>(G_i) ,G_u)$. These can be done in $O(mn)$ time. Without loss of generality, assume $y_1\geq y_2\cdots \geq y_n$ and $\Lambda(\calP^>(G_1), G_u)\geq \Lambda(\calP^>(G_2), G_u)\geq \cdots\geq\Lambda(\calP^>(G_s), G_u)$. We then locate $c^*_1$ and $c^*_2$ for each following case. 

    \begin{enumerate} 
    \item The median of some uncertain point $P_j$ with $y_j = y_1$ or $y_j = y_2$ is at $G_u$. By Lemma~\ref{lem:articulatemedian}, $G_u$ is the median of an uncertain point if its probability sum in any split subgraph of $G_u$ is not greater than $0.5$. Due to $\lambda^*\geq\max_{i=1}^{n}\{w_i\Ed(P_i,x^*_i)\}$, if $y_j = y_1$ then $\lambda^* = y_1$ and otherwise, $y_1> y_2 = y_j$ and $\lambda^* = \max\{w_1\Ed(P_1,x^*_1), y_2\}$.  
    
    The subset of all uncertain points whose medians are at $G_u$ is subset $\calP-\cup_{i=1}^{s}\calP^>(G_i)$ and hence, it can be found in $O(n)$ time. Additionally, if setting the weight of each $P_i$ with $i>1$ as zero, then $x^*_1$ is exactly the center w.r.t. $\calP$ on $G$ and the optimal objective value is $w_1\Ed(P_1,x^*_1)$. Consequently, we can decide the existence of such an uncertain point, and if it exists, compute $\lambda^*$ totally in $O(mn\log mn)$ time. 

    \item A split subgraph $G_j$ of $G_u$ with $2<j\leq t$ has $\Lambda(\calP^>(G_1),G_u)\geq\Lambda(\calP^>(G_j),G_u)\geq \Lambda(\calP^>(G_2),G_u)$. Lemmas~\ref{lem:articulatemedian} and~\ref{lem:cyclemedian} imply the following. Every hanging subgraph $G_i\cup\{G_u\}$ with $\Lambda(\calP^>(G_i),G_u)<\Lambda(\calP^>(G_2), G_u)$ except for $G_u$ contains neither $q^*_1$ nor $q^*_2$ since otherwise, the center can be moved to $G_u$ without increasing the objective value. Furthermore, if $\Lambda(\calP^>(G_1),G_u) = \Lambda(\calP^>(G_2),G_u)$ then $\lambda^* = \Lambda(\calP^>(G_1),G_u)$. Otherwise, $\Lambda(\calP^>(G_1),G_u) > \Lambda(\calP^>(G_2),G_u) =\Lambda(\calP^>(G_j),G_u)$. It follows $\lambda^*\geq\Lambda(\calP^>(G_2),G_u)$. 
    
    We decide the feasibility of $\Lambda(\calP^>(G_2),G_u)$ in $O(m^2n^2)$ time. If it is feasible, then $\lambda^* = \Lambda(\calP^>(G_2),G_u)$. If not, in order to obtain a smaller objective value, all necessary centers must be placed in the hanging subgraph $G_1\cup\{G_u\}$. Hence, $c^*_1$ and $c^*_2$ are all in $u$'s hanging subtree on $T$ representing 
    $G_1\cup\{G_u\}$. We thus return $u$'s adjacent node in that split subtree. 
    
    \item Every split subgraph $G_j$ of $G_u$ with $2<j\leq t$ has $\Lambda(\calP^>(G_j),G_u)<\Lambda(\calP^>(G_2),G_u)$. Clearly, $c^*_1$ is in $u$'s hanging subtree representing $G_1\cup\{G_u\}$. However, center $q^*_2$ might be in $G_1$ or $G_2$. Apply the decision algorithm to decide the feasibility of $\Lambda(\calP^>(G_2),G_u)$. If $\Lambda(\calP^>(G_2),G_u)\geq\lambda^*$, then $q^*_2$ might be in $G_2$, and otherwise, $q^*_2$ can be placed in $G_1\cup\{G_u\}$ in order to cover more uncertain points under $\lambda^*$. As a result, the two adjacent nodes of $u$ in its hanging subtrees for $G_1\cup\{G_u\}$ and $G_2\cup\{G_u\}$ are returned. 
    \end{enumerate}

    Handling each above case takes $O(m^2n^2)$ time in the worst case, which dominates the running time. Hence, the lemma holds.\qed
\end{proof}

\begin{lemma}\label{lem:CriticalCycleNode}
    Given a cycle node $u$ on $T$, we can decide in $O(m^2n^2)$ time whether $u$ is a critical node and if not, which $H$-subtrees of $u$ on $T$ contain $c^*_1$ and $c^*_2$. 
\end{lemma}
\begin{proof}
    For each split subgraph $G_i$ of $G_u$ with $1\leq i\leq s$, we compute in $O(mn)$ time value $\Lambda(\calP^>(G_i),h_i)$. Suppose $\Lambda(\calP^>(G_1),h_1)\geq\Lambda(\calP^>(G_2),h_2)\geq\cdots\geq \Lambda(\calP^>(G_s),h_s)$. By Lemmas~\ref{lem:articulatemedian} and~\ref{lem:cyclemedian}, $q^*_1$ and $q^*_2$ are not likely to be in any subgraph $H_j/\{h_j\}$ where $H_j$ is $G_j\cup\{h_j\}$ and that has $\Lambda(\calP^>(G_j),h_j)<\Lambda(\calP^>(G_2),h_2)$. 
    
    On the one hand, $\Lambda(\calP^>(G_1),h_1)\geq\Lambda(\calP^>(G_2),h_2) = \Lambda(\calP^>(G_3),h_3)$. Clearly, for any $1\leq j\leq s$, subgraph $H_j/\{h_j\}$ with $\Lambda(\calP^>(G_j),h_j) = \Lambda(\calP^>(G_2),h_2)$ contains neither $q^*_1$ nor $q^*_2$. Thus, $\lambda^*\geq\Lambda(\calP^>(G_2),h_2)$. 

    If $\Lambda(\calP^>(G_1),h_1)=\Lambda(\calP^>(G_2),h_2)$ then all necessary centers must be on $G_u$. Hence, $c^*_1 = c^*_2 =u$. Otherwise, we apply Lemma~\ref{lem:CriticalNonCycleNode} to $h_1$ to determine which of subgraphs $H_1$ and $G_u$ contains a critical edge or cycle. Generally, only node $u$ or both $u$ and a node $u_1$ adjacent to $h_1$'s hinge node in $T^H_u$'s split subtree for $H_1$ are obtained. 
    For the former case, all necessary centers are on $G_u$ and thereby $c^*_1 = c^*_2 =u$. In the later case, $u$ is a critical node and the hanging subtree of $T^H_u$ representing $H_1$ contains the other. 
    
    On the other hand, $\Lambda(\calP^>(G_1),h_1)\geq\Lambda(\calP^>(G_2),h_2) > \Lambda(\calP^>(G_3),h_3)$. If $h_1\neq h_2$, we first apply Lemma~\ref{lem:CriticalNonCycleNode} to $h_1$ to decide whether or not $H_1$ contains a critical edge or cycle. Generally, some node(s) are obtained. 
    If only $u$ is returned, due to $\Lambda(\calP^>(G_1),h_1)\geq\Lambda(\calP^>(G_2),h_2)$, all centers must be on $G_u$, which means $c^*_1 = c^*_2 =u$. When only node $u_1$ is obtained, $H_1$ contains both $c^*_1$ and $c^*_2$. For the case where two nodes are obtained, one must be $u$ and the other is $u_1$. Definitely, $H_1$ contains a critical cycle or edge. But we need to further decide whether $G_u$ is the other critical cycle or $H_2$ contains the other, which can be addressed by applying Lemma~\ref{lem:CriticalNonCycleNode} to $h_2$. As long as we receive the node that adjacent to $h_2$'s hinge node in $T^H_u$'s split subtree for $G_2$, $H_2$ contains the other critical cycle or edge but $G_u$ does not. Otherwise, $G_u$ is a critical cycle.  

    For the case where $h_1$ and $h_2$ are same hinge, we determine which of subgraphs $H_1$, $H_2$, and $G_u$ contain a critical edge or cycle by applying Lemma~\ref{lem:CriticalNonCycleNode} only to $h_1$. 
    The remaining is similar to the above case. We omit the details.  

    The running time for handling each above case is dominated by the time of the $O(1)$ calls on the decision algorithm. So, the time complexity is $O(m^2n^2)$.\qed
\end{proof}

In general, $c^*_1$ is obtained after performing $O(\log mn)$ recursive steps on $T$ where each recursive step takes $O(m^2n^2)$ time. After obtaining $c^*_1$, $c^*_2$ are computed in a similar way in $O(m^2n^2\log mn)$ time. With $c^*_1$ and $c^*_2$, the following lemma is adapted to compute $\lambda^*$ in $O(m^2n^2\log mn)$ time. 

\begin{lemma}\label{lem:basecase}
Given $c^*_1$ and $c^*_2$, we can compute $\lambda^*$ in $O(m^2n^2\log mn)$ time.
\end{lemma}
\begin{proof}
We observe that $\lambda^*$ belongs to the set consisting of values $w_i\Ed(P_i,x^*_i)$ of each $P_i\in\calP$ and the $y$-coordinates of all intersections between functions $y = w_i\Ed(P_i,x)$ w.r.t. every edge on $G$ where the slopes of two functions are of opposite signs. This is because the one of $q^*_1$ and $q^*_2$ leading $\lambda^*$ 
must be the center of all uncertain points covered by it under $\lambda^*$ and the one-center objective value of $\calP$ on $G$ belongs to such a set~\cite{ref:HuCo23}. 

To find $\lambda^*$, it suffices to find such sets of candidate values w.r.t. every edge in subgraphs of $G$ represented by $c^*_1$ and $c^*_2$. Since $c^*_1$ and $c^*_2$ are out-of-cycle edges or cycles on $G$, it takes $O(mn)$ time to determine all functions $y=w_i\Ed(P_i,x)$ w.r.t. $x\in c^*_1$ and $x\in c^*_2$, respectively. 
These generate $O(n)$ piece-wise linear functions of complexity $O(m)$ in the $x,y$-coordinate plane where all edges of $c^*_1$ and $c^*_2$ are on $x$-axis. 

If $c^*_1$ is a cycle, then that candidate set w.r.t. all edges on $c^*_1$ belongs to the set of the $y$-coordinates of the intersections between the $O(mn)$ lines that contain the line segments on the graphs of all $y=w_i\Ed(P_i,x)$. Note that this set also includes all values $w_i\Ed(P_i,x^*_i)$ w.r.t. $c^*_1$. 

Otherwise, $c^*_1$ is an out-of-cycle edge on $G$. Each $y=w_i\Ed(P_i,x)$ is a line segment in the $x,y$-coordinate plane. The candidate set w.r.t. $c^*_1$ belongs to the set that includes not only the $y$-coordinates of all intersections between the $O(n)$ lines containing line segments of all $y=w_i\Ed(P_i,x)$ but also the $y$-coordinates of their intersections with the two vertical lines through $c^*_1$'s incident vertices on $x$-axis. The latter involves all values $w_i\Ed(P_i,x^*_i)$ for uncertain points whose medians are at $c^*_1$'s incident vertices. 

Denote by $L$ the set of these $O(mn)$ lines caused by all $y=w_i\Ed(P_i,x)$ respectively w.r.t. $c^*_1$ and $c^*_2$ 
including the two vertical lines through their incident vertices 
if they are out-of-cycle edges. Based on the above analysis, computing $\lambda^*$ is equivalent to finding the smallest feasible $y$-coordinate among all intersections of lines in $L$, 
that is, finding the lowest vertex with the smallest feasible $y$-coordinate in the line arrangement of $L$. Thus, the line arrangement search technique~\cite{ref:ChenAn13} can be adapted to find $\lambda^*$ 
by using our decision algorithm for the feasibility test, which runs in $O(m^2n^2\log mn)$ time. Therefore, the time complexity for computing $\lambda^*$ is $O(m^2n^2\log mn)$.\qed
\end{proof}

Recall that at the beginning we reduced the problem into a vertex-constrained instance by Lemma~\ref{lem:twocenterreduction} in $O(|G|+ mn)$ time. Combining all above efforts, the following result is derived.

\begin{theorem}\label{theorem:1}
The two-center problem of $n$ uncertain points on a cactus graph $G$ 
can be solved in $O(|G|+ m^2n^2\log mn)$ time.
\end{theorem}

\section{Conclusion}
As mentioned in the previous work, no algorithm is known for the general $k$-center problem w.r.t. $\calP$ on $G$. It is thus natural to ask whether our algorithm can be extended to solve the $k$-center problem for $k$ being a constant. Unfortunately, even for $k=3$, Lemmas~\ref{lem:centerdetectingarticulate} and~\ref{lem:CriticalNonCycleNode} cannot be generalized since at any articulate point of $G$, when exactly two hanging subgraphs are known to contain all centers, it is hard to figure out which of them must contain two centers even for the decision problem. 

Moreover, for $G$ be a tree, the decision $k$-center problem w.r.t. $\calP$ can be solved by a greedy algorithm that places least centers on $G$ in the bottom-up manner to cover $\calP$ under any given $\lambda$. However, for $G$ being a cactus graph, this greedy scheme does not work since $\Ed(P_i,x)$ for $x$ being in any path is no longer a convex function. It would be interesting to find out whether the general $k$-center problem can be solved efficiently. 


\end{document}